\documentclass[11pt]{article}

\usepackage{amsmath,amssymb,amsfonts,amsthm} 
\usepackage{mathrsfs,latexsym} 
\usepackage{mathtools}

\newtheorem{theorem}{Theorem}[section]

\newtheorem{lemma}[theorem]{Lemma}
\newtheorem{proposition}[theorem]{Proposition}

\numberwithin{equation}{section}

\DeclareMathAlphabet{\pazocal}{OMS}{zplm}{m}{n}

\usepackage[mathscr]{eucal} 

\usepackage{color}
\usepackage{cite}
\usepackage{epsfig}
\usepackage{graphicx}

\usepackage{bm}
\usepackage{cite}

\def\CC{\mathbb{C}}
\def\NN{\mathbb{N}}
\def\RR{\mathbb{R}}

\newcommand{\cD}{{\mathcal{D}}}
\newcommand{\cF}{{\mathcal{F}}}

\newcommand{\cS}{{\mathcal{S}}}

\newcommand{\bp}{{\mbox{\boldmath $p$}}}

\newcommand{\bx}{{\mbox{\boldmath $x$}}}

\newcommand{\by}{{\mbox{\boldmath $y$}}}

\def\bO{\mbox{\boldmath $O$}}

\newcommand{\fA}{{\mathfrak A}}

\newcommand{\fR}{{\mathfrak R}}

\newcommand{\bpartial}{\bm{\partial}}

\def\eg{{\it e.g.\ }}
\def\viz{{\it viz.}}
\def\ie{{\it i.e.\ }}
\def\etc{{\it etc}}

\def\be{\begin{equation}} 
\def\ee{\end{equation}}

\newcommand{\vertiii}[1]{{\left\vert\kern-0.25ex \left\vert\kern-0.25ex
    \left\vert #1 
    \right\vert\kern-0.25ex\right\vert\kern-0.25ex\right\vert}}

\begin{document} 

\title{\hspace*{-5mm}
  \mbox{Many-body physics and resolvent algebras}}
\author{Detlev Buchholz${}^{(a)}$ \ and \ Jakob Yngvason${}^{(b)}$ \\[2mm]
  \small ${}^{(a)}$ Mathematisches Institut, Universit\"at G\"ottingen,
  37073 G\"ottingen, Germany \\
 \small  ${}^{(b)}$  Fakult\"at f\"ur  Physik, Universit\"at Wien, 
 1090 Wien, Austria \\[4mm]}
\date{}

\maketitle

\vspace*{-4mm}
\noindent \textbf{Abstract.}
Some advantages of the algebraic approach to many body physics,
based on resolvent algebras, are illustrated by the 
simple example of non-interacting bosons which are confined in
compact regions with soft boundaries. It is shown that the dynamics of
these systems converges to the spatially homogeneous 
dynamics for increasing regions and particle numbers and a
variety of boundary forces. The corresponding correlation functions
of thermal equilibrium states also converge in this limit.
Adding to the regions further particles, the limits are steady states,
including Bose-Einstein condensates. They  
can either be spatially homogeneous, or they are inhomogeneous
with varying, but finite local particle densities. In case of
this spontaneous breakdown of the spatial symmetry, the
presence of condensates can be established by exhibiting
temporal correlations over large
temporal distances (memory effects).



\section{Introduction}
\label{sec1}
\setcounter{equation}{0}

The original formulation of quantum mechanics is based
on observables and their temporal evolution (Heisenberg picture). 
Yet this algebraic approach is rarely used in many-body physics.
The reason for this is the fact that the common algebras of the 
canonical observables are not stable under the action of 
dynamics that describe interaction processes;
this applies both to the polynomial algebra of the position and
momentum operators as well as the Weyl algebra formed by 
unitary operators generated by them.
Instead, one makes use of the fact that for systems of any finite 
number of particles these algebras have unique representations on Hilbert
spaces  (Stone-von Neumann theorem). There the generators of the
dynamics can be represented by self-adjoint operators, which are 
accessible to mathematical analysis.

\medskip 
The limitations of this approach
only become apparent when one moves on to the physically desirable
idealization of systems with an infinite number of particles.
This step is needed in order
to precisely describe different phases, Bose-Einstein
condensation, or the phenomenon
of spontaneous symmetry breaking. In this step, one leaves
the scope of the Stone-von Neumann theorem and is confronted with
a multitude of inequivalent representations of the canonical 
observables which differ in temperature, chemical potential, 
spatial structures, \etc. The corresponding representations
must be determined in each individual case. On the other hand, 
the algebraic
properties of the underlying observables and their dynamical
evolution remain unchanged. So these algebraic structures are   
the common core of the multitude of states
that describe very different physical systems for a given
dynamics. 

\medskip 
It is therefore gratifying that an algebraic formalism is now
available that describes both finite and infinite particle systems
and their dynamics. As will be explained, it consists of a single algebra,
the resolvent algebra, on which the dynamics acts by automorphisms
\cite{BuGr}. The observables are  
(gauge) invariant under the adjoint action of the unitary
group generated by the particle number operator. They constitute
a subalgebra of the resolvent algebra which is represented by bounded
operators on Hilbert space for any finite or infinite particle
number \cite{Bu2}. 
These spaces can be reconstructed from expectation functionals
on the algebra by the Gelfand-Naimark-Segal (GNS)
construction \cite[Ch.~III.2.2]{Ha}. 
By considering sequences of functionals one can
change the representation and thereby arrive in the 
limit at functionals describing the desired physical
systems. This algebraic approach has already proved to be useful.
For example, it entered in the construction of ground and
equilibrium states for oscillating lattice systems with
nearest neighbor interactions \cite{Bu1} and it shed new
light on the characterization of Bose-Einstein condensates
\cite{BaBu, Bu3, Bu4}.

\medskip 
It is the aim of the present
article to illustrate the virtues of the algebraic approach by
discussing several problems of physical interest
for non-interacting bosons. In a first step, we consider the dynamics 
of finite numbers of bosons which are confined to compact 
spatial regions with soft (harmonic) boundary forces.  
The action of the corresponding dynamics on the 
observables converges to the spatially homogeneous dynamics 
under mild conditions on the strength of the confining forces
for increasing particle numbers and regions. In other words,
the dynamics of the confined system has a unique thermodynamic
limit which does not depend sensitively on the  
confining forces. We then consider sequences of
thermal equilibrium states (Gibbs ensembles) on the
observable algebra with increasing particle numbers and
establish the convergence of the corresponding expectation
functionals. The computations greatly
simplify by making use of the
Kubo-Martin-Schwinger (KMS) condition for the expectation
values. The limits are the well-known homogeneous
equilibrium states, including Bose-Einstein condensates. 
Although the limit dynamics is spatially homogeneous, this does
not imply, however, that the limits of thermal equilibrium
states have to be homogeneous. We exemplify
this spontaneous breakdown of symmetry using an example in
which the equilibrium states in limited regions are “overfed” with
particles whilst keeping the local particle densities finite.
The condensate then escapes to more distant regions, leaving
an inhomogeneous equilibrium state in the thermodynamic limit. 
There, it no longer makes sense to determine the presence of
condensates by examining spatial correlations. Instead, we look
at the temporal correlations in these stationary states, which
persist over large time intervals (memory effects). 

\medskip
Our article is organized as follows: The subsequent 
section summarizes some basic definitions and 
results for the resolvent algebra.
In Sect. 3 the dynamics for non-interacting systems in finite
regions with harmonic boundary forces is introduced, and it is shown
to converge to the spatially homogeneous free dynamics in the
thermodynamic limit. 
The thermodynamic limit of the corresponding
thermal equilibrium states is established in Sect.~4.
Condensates are considered in Sect.~5 and
the spontaneous breakdown of the translation symmetry
of condensates in
the thermodynamic limit is the topic of Sect.~6.
The article concludes with a short summary and an outlook
on interacting systems. 

\section{Resolvent algebra and dynamics}
\label{sec2}
\setcounter{equation}{0}

Since we are dealing with systems of an arbitrary finite number of
bosons, we make use of Fock space and the canonical creation
and annihilation operators acting on it. Fock space $\cF$ is spanned
by the $n$-fold symmetrized tensor products of the single particle 
states in $L^2(\RR^s)$, where $n \in \NN_0$ is
the particle number and $s$ is the dimension of
configuration space.
The creation and annihilation operators are denoted by
$a^*(f)$, respectively $a(g)$, where $f,g \in \cD(\RR^s)$.
Denoting by $\langle \, \cdot \, , \, \cdot \, \rangle$ the
scalar product in $L^2(\RR^s)$,
they satisfy the commutation relations
\be
   [a(g), a^*(f)] = \langle g, f \rangle 1 \, , \quad
    [a^*(g), a^*(f)] = [a(g), a(f)] = 0 \, .
\ee
We restrict these operators to 
the dense subspace $\cD(\RR^s) \subset L^2(\RR^s)$ of test functions
with compact support in order to cope with the distributional
properties of states which can appear in the limit of large particle numbers.

\medskip 
The resolvent algebra, encoding all properties of the
canonical operators, can be defined by algebraic relations 
and is faithfully represented on Fock space $\cF$ \cite{BuGr}.
Its building blocks are the resolvents 
\be \label{e.2.2} 
R(\lambda, f) \coloneqq (\lambda + a^*(f) + a(f))^{-1} \, ,
\quad \lambda \in \CC \backslash \RR \, , \ f \in \cD(\RR^s) \, .
\ee
They form an algebra of bounded operators on $\cF$ which is stable
under taking adjoints. Its closure with regard to the operator
norm on $\cF$ is a C*-algebra~$\fR$, the resolvent algebra.

\medskip
The resolvent algebra includes the observables which are of
interest in many body physics, such as $n$-particle
density operators. Since these do not change particle
numbers, they are fixed points under
the adjoint action of the unitaries $\mu \mapsto e^{i \mu N}$,
where~$N$ is the particle number operator on $\cF$. 
We refer to this action as gauge transformation. It will
be shown in Sect.~3 that the closed 
subalgebra $\fA \subset \fR$ of gauge-invariant operators,
the observable algebra, is generated by the resolvents 
\be \label{e.2.3} 
A(\lambda,f) \coloneqq (\lambda + a^*(f) a(f))^{-1} \, , \quad \lambda > 0 \, , \
f \in \cD(\RR^s) \, .
\ee
This algebra is also faithfully represented on $\cF$. It
leaves each $n$-particle subspace $\cF_n \subset \cF$
invariant and acts irreducibly on it, $n \in \NN_0$. The properties of the
algebra $\fA$ have been unraveled in \cite{Bu2}. What matters here
is the fact that the restricted maps  
\be
f \mapsto A(\lambda,f) \upharpoonright \cF_n \, , \quad f \in \cD(\RR^s) \, , 
\ee
extend continuously in the norm topology on $\cF_n$ to all 
$f \in L^2(\RR^s)$ for any $n \in \NN_0$. This implies that the limits
are represented by elements of~$\fA$ which may, however, depend
on the chosen representation space. Even though the operators
$A(\lambda,f)$ are not elements of $\fA$ for arbitrary square integrable
functions $f$,  we can make use of them in the following because
of this observation. Technically,
they are elements of the projective limit of the
family of represented algebras~$\fA$ for $n \in \NN_0$,
cf.~\cite{Bu2}. 

\medskip
The Hamiltonians of interest here have, on a dense domain
in~$\cF$, the form 
\be
H = \int \! d\bx \, \big(\bpartial a^*(\bx) \bpartial a(\bx)
+ ( c^2 \bx^2 + U(\bx) + c' ) \, a^*(\bx) a(\bx) \big) \, .
\ee
Here $\bpartial$ denotes the gradient and $U$ is any continuous
potential which vanishes at infinity. In order to cover also 
confined (trapped) systems, we include a harmonic
potential of arbitrary strength $c^2$;
the constant shift $c'$ may be interpreted as a 
chemical potential. Other regular confining
forces can be admitted, cf.\  \cite[Sect.~6]{Bu3},
but they will not be considered here. 
The adjoint action of the corresponding
unitary time translations on the observable algebra (the dynamics)
is denoted by $t \mapsto \alpha(t) \coloneqq \mbox{Ad} \, e^{itH}$, $t \in \RR$. 
One has for any $n \in \NN_0$
\be
\alpha(t)(A) \upharpoonright \cF_n \in \fA \upharpoonright \cF_n \, ,
\quad A \in \fA \, ,
\ee
\ie the dynamics preserves
the represented algebras and their projective
limit. Moreover, 
$t \mapsto \alpha(t)$ acts pointwise norm continuously
on these algebras. It is worth noting that the
dynamics retains these properties even when the Hamiltonian
includes interactions with continuous two-body potentials
that vanish at infinity, cf.\ \cite{Bu2} and the conclusions. 

\section{Thermodynamic limit of dynamics}
\label{sec3}
\setcounter{equation}{0}  

We consider now the dynamics of trapped bosons which
are confined in balls of radius $R>0$ by soft (harmonic) boundary forces.
The corresponding Hamiltonians are 
\be \label{e.3.1} 
H_R \coloneqq \int \! d\bx \, \bpartial a^*(\bx) \bpartial a(\bx)
+ c^2(R)  \int \! d\bx \, \theta(\bx^2 - R^2) (\bx^2 - R^2)
\, a^*(\bx) a(\bx) \big) \, ,
\ee
where $\theta$ is the Heaviside step function. The strength of
the repulsive forces at the boundary can be adjusted by
choosing $c(R)$. These Hamiltonians coincide with the
free Hamiltonian in the interior of the ball and are of the type 
introduced in the previous section. This is because the 
potential is obtained by adding to the shifted harmonic potential
$\bx \mapsto c(R)^2( \bx^2 - R^2)$ the continuous function 
$\bx \mapsto c(R)^2 \theta(R^2 - \bx^2)(R^2 - \bx^2)$, 
which vanishes outside the ball. 
          
\medskip
We will see that these dynamics approximate in the limit of large $R$
the free dynamics without repulsive boundary forces for a large 
choice of constants  $c(R)$. In accordance with the chosen
notation, the resulting Hamiltonian is denoted by $H_\infty$
(no forces at finite $R$). The following lemma is a key
ingredient in this analysis.
\begin{lemma} \label{l.3.1}
  Let $R \mapsto c(R)$ be polynomially bounded in the limit of large 
  $R$. Then, for any test function 
  $f \in \cD(\RR^s)$, one has
  \be
  \lim_{R \rightarrow \infty} \| (e^{-itH_R} - e^{-itH_\infty})f \|_2 = 0 \, ,
  \quad t \in \RR \, ,
  \ee
  where $\| \, \cdot \, \|_2$ denotes the norm on $L^2(\RR^s)$. 
    The limit is reached faster than any inverse power of $R$, uniformly for
    $t$ in compact sets. 
\end{lemma}  
\begin{proof}
By standard arguments one obtains the estimate 
  \be \label{e.3}
  \|  (e^{-itH_R} - e^{-itH_\infty})f \|_2 \leq
  \int_0^t \! du \, \|(H_R - H_\infty) e^{-iuH_\infty})f \|_2 \, .
  \ee
  Now
  \be \label{e.4} 
  \|(H_R - H_\infty) \, e^{-iuH_\infty}f \|_2^2 =
  c(R)^4 \int_{|\bx| \geq R} \! d\bx \, (\bx^2 - R^2)^2 \, 
  |(e^{-iuH_\infty} f)(\bx)|^2 , 
  \ee
  and, for any $m \in \NN$, one has
  \be
(1 + \bx^2)^m \, |(e^{-iuH_\infty} f)(\bx)| \leq (1 + u^{2m}) \, p_m(f) \, ,
  \ee
  where $p_m(f)$ is some Schwartz seminorm of $f$. The latter estimate
  follows from the equality
  \be
  (1 + \bx^2)^m \, (e^{-iuH_\infty} f)(\bx)
  = (2 \pi)^{-m/2} \int \! d\bp \, e^{i \bp \bx} (1 - \Delta_{\bp})^m
  e^{-iu\bp^2} \widetilde{f}(\bp) \, ,
  \ee
  which is obtained by partial integration in the Fourier integral,
  involving   
  the Laplace operator $\Delta_{\bp}$ in momentum space. Thus the integral in
  equation~\eqref{e.4} decreases faster than any inverse power of $R$,
  uniformly for $u$ varying in compact sets. So  
  the statement is obtained from the estimate~\eqref{e.3} by integration.
  \end{proof}  
It follows from this lemma that the unitaries
$e^{itH_R}$ converge on Fock space~$\cF$ in the limit of
large $R$ to $e^{itH_\infty}$
in the strong operator topology, $t \in \RR$. But this
result is not sufficient for our purposes. We need to
determine the properties of the corresponding automorphisms
on the observable algebra $\fA$, where the following technical
lemma greatly simplifies the analysis.
\begin{lemma} \label{l.3.2} 
  The linear span of all resolvents
  $A(\lambda,f) = (\lambda + a^*(f)a(f))^{-1}$ with $f \in \cD(\RR^s)$ and
  $\lambda > 0$ is norm-dense in $\fA$. 
\end{lemma}
\begin{proof}
  The linear span of all basic resolvents
  $R(\lambda,f) = (i \lambda + a^*(f) + a(f))^{-1}$,
  where \mbox{$f \in \cD(\RR^s)$}
  and $\lambda \in \RR \backslash \{ 0 \}$, 
  is norm-dense in the full resolvent algebra $\fR$, cf.\ \cite{BuNu}.
  Hence, the linear span of their means over the
  gauge transformations is norm-dense in the observable
  algebra $\fA$. Given $f \neq 0$, the mean of the 
  basic resolvents $R(\lambda, f)$ are continuous functions
  vanishing at infinity of the
  corresponding number operator
  $N(f) \coloneqq \|f \|_2^{-2} a^*(f) a(f)$;
  they generate an abelian subalgebra of $\fA$. Since the
  span of the resolvents given in the lemma is norm-dense 
  in that algebra, the statement follows. \end{proof}  
We consider now the restriction of the algebra of observables $\fA$ to 
the subspaces $\cF_n$ of Fock space
with particle number $n \in \NN$. The spatial
region occupied by the particles is controlled 
by the dynamics $\alpha_R$ fixed by the Hamiltonians
$H_R$, $R > 0$. As before, we restrict 
our attention to coupling functions $R \mapsto c(R)$ in
$H_R$ which are polynomially bounded.
Denoting by $\| \cdot \|_n$ the operator norm on $\cF_n$, the
following result obtains.

\begin{lemma} \label{l.3.3}
  Let $R_n$, $n \in \NN$, be an increasing sequence of radii
  such that $\lim_{n \rightarrow \infty} n \, R_n^{-k} = 0$ for some $k > 0$. Then
  \be
  \lim_{n \rightarrow \infty}
  \| \alpha_{R_n}(t)(A) - \alpha_\infty(t)(A) \|_n = 0 \quad \text{for}
  \quad A \in \fA \, ,
  \ee
  uniformly in $t$ in compact sets. 
\end{lemma}  

\noindent
\textbf{Remark:} The constraints on 
the radii $R_n$ of the regions, depending on
the particle number~$n$, cover the familiar  
assumption that the mean particle density $n R_n^{-s}$ is held fixed
for increasing $n$. We will make other choices later. 

\begin{proof}
  Since the action of automorphisms on
  the algebra is continuous in norm, it suffices
  to establish the statement for a dense subalgebra of $\fA$.
  Hence, in view of Lemma~\ref{l.3.2}, we need to consider only 
  the total set of resolvents
  ${A}(\lambda,f)$ with $f \in \cD(\RR)$, $\lambda > 0$.
  Now
  $\alpha_R(t)(A(\lambda,f)) = A(\lambda, e^{itH_R}f)$  
  for $R \in \RR \bigcup {\infty}$. Thus, denoting the particle
  number operator on $\cF$ by $N$, one has for $n \geq 1$
  \begin{align}
  & \| \alpha_{R_n}(t)(A(\lambda,f)) -
    \alpha_{\infty}(t)(A(\lambda,f)) \|_n \nonumber \\ 
  & = n  \, \| N^{-1/2}
  (A(\lambda,e^{itH_{R_n}}f)  - A(\lambda,e^{itH_\infty}f)) N^{-1/2}  \|_n \nonumber \\
  & \leq  n \lambda^{-2} \,  \| N^{-1/2}
  \big(a^*(e^{itH_{R_n}}f) a(e^{itH_{R_n}}f) -
  a^*(e^{itH_{\infty}}f) a(e^{itH_{\infty}}f) \big) N^{-1/2} \|_n \nonumber \\
  & \leq 2n \lambda^{-2}  \| a(e^{itH_\infty}f) N^{-1/2} \|_n 
  \, \|a((e^{itH_{R_n}} - e^{itH_\infty})f) N^{-1/2} \|_n \, . 
  \end{align}
  Using the well-known estimate $\| a(g) N^{-1/2} \|_n \leq \| g \|_2$
  for $g \in L^2(\RR^s)$, one arrives at the bound 
  \be
   \| \alpha_{R_n}(t)(A(\lambda,f)) -
   \alpha_{\infty}(t)(A(\lambda,f)) \|_n \leq 2n \lambda^{-2} \| f \|_2 \, 
   \| (e^{it H_{R_n}} - e^{it H_\infty})f \|_2 \, .
  \ee
   The statement then follows from Lemma \ref{l.3.1}
\end{proof}  

This lemma implies that 
the sequence of dynamics $\alpha_{R_n}$, $n \in \NN_0$, converges
pointwise on $\fA$ to $\alpha_\infty$ in the operator norm
on any subspace of $\cF$ with maximal particle number
$l \in \NN_0$. For operators $B$ which 
commute with the particle number operator, this 
norm is given by 
$\vertiii{B}_l := \sup \sum_{k = 0}^l c_k \| B \|_k$, where  
the supremum is taken over all non-negative coefficients
$c_k$, $k = 0, \dots , l$, which sum up to $1$. 
For elements $A \in \fA$ the norms $k \mapsto \| A \|_k$
are monotonically increasing, cf.\ \cite[Lem.\ 3.4]{Bu2},
so one has in this case $\vertiii{A}_l = \|A\|_l$, $l \in \NN_0$, 
This leads to the following result. 

\begin{proposition} \label{p.3.4}
  Let $\alpha_{R_n}$, $n \in \NN$,  be a sequence of
  dynamics with radii satisfying the
  condition given in Lemma \ref{l.3.3} and
  let $l \in \NN_0$. Then, for any $A \in \fA$, 
  \be
  \lim_{n \rightarrow \infty} 
  \vertiii{\alpha_{R_n}(t)(A) - \alpha_\infty(t)(A)}_l = 0 \, 
  \ee
  uniformly in $t$ in compact sets. 
\end{proposition}  
\begin{proof}
  Given $l$, let $n \geq l$. Then
  \be 
  \vertiii{\alpha_{R_n}(t)(A) - \alpha_\infty(t)(A)}_l 
  \leq \| \alpha_{R_n}(t)(A) - \alpha_\infty(t)(A) \|_{n} \, ,
  \ee
  so the statement follows from Lemma \ref{l.3.3}.
  \end{proof}  

It follows from this proposition that all weak limit points of
the family $\Sigma_n$ of states with not more than $n$
particles, which are stationary under the action of
$\alpha_{R_n}$, $n \in \NN$, are stationary under the action of
the limit dynamics $\alpha_\infty$. Since 
$\alpha_\infty$ leaves $\fA$ invariant on all subspaces
with limited particle number, 
it follows that in the GNS-representation
induced by the limit states the dynamics is
implemented by the adjoint action of
unitary operators $t \mapsto V_\infty(t)$.
These differ in general substantially (\eg with
regard to spectral properties) from the exponential
function of the Hamiltonian $H_\infty$ on Fock space, 
but the time evolution of the observables remains unchanged.

\section{Equilibrium states}
\label{sec4}
\setcounter{equation}{0}  

The expectation values of observables in thermal equilibrium
states~$\omega_\beta$, described by Gibbs ensembles, are given by 
\be
\omega_{\beta}(A) \coloneqq \mbox{Tr} \, e^{-\beta H} A \, / \, 
\mbox{Tr} \, e^{-\beta H} \, , \quad A \in \fA \, .
\ee
Here $\beta$ is the inverse temperature
and $H$ is the Hamiltonian to which
other constants of motion may have been added, 
such as multiples of the particle number
operator. The implicit 
condition underlying this formula is the
assumption that the exponential function of
$H$ has a finite trace for $\beta > 0$. 
Depending on the ensembles which one
considers (canonical or grand canonical), the trace is to be taken
over the $n$-particle subspaces $\cF_n \subset \cF$ or over Fock space~$\cF$.
The condition on the trace is generally satisfied for trapped
systems with sufficiently strong confining
forces, but it fails in the thermodynamic 
limit. What remains in this limit are
characteristic properties of the
correlation functions inherited from
the approximating states. They satisfy the KMS-condition
\cite[Ch.\ V]{Ha}, which formally reads 
\be
t \mapsto \omega_\beta(A \, \alpha(t + i \beta)(B))
= \omega_\beta(\alpha(t)(B) \, A) \, , \quad A, B \in \fA \, . 
\ee
It is to be understood as a relation between the 
boundary values of analytic extensions 
of the correlation functions to the
strip $\{ z \in \CC : 0 \leq \mbox{Im} z \leq \beta \}$.

\medskip
The KMS-condition is not only a characteristic feature 
of thermal equilibrium states, but it is also useful in
computations. We briefly recall this fact in case of a  
dynamics describing particles which do not interact with
each other. Using an obvious notation, it acts 
on the creation operators by 
\be
t \mapsto \alpha(t)(a^*(f)) = a^*(e^{itH} f) \, , \quad f \in \cD(\RR^s) \, .
\ee
Assuming that $H$ is positive, one obtains the following 
equalities, making use of the KMS-condition
for creation and annihilation operators and their   
commutation relations,  
\be
\omega_\beta(a^*(f) a(g)) =
\omega_\beta(a(g) a^*(e^{- \beta H} f)) =
\langle g, e^{- \beta H} f \rangle  + \omega_\beta(a^*(e^{- \beta H} f) a(g)) \, .
\ee
Applying the inverse of $(1 - e^{-\beta H})$ to $f$,  
one arrives at the familiar one-particle density
matrix 
\be
\omega_\beta(a^*(f) a(g)) =  \langle g, (e^{\beta H} - 1)^{-1}  f \rangle
\, , \quad f,g \in \cD(\RR^s) \, .
\ee

The computations of the
higher density matrices in the canonical ensembles are a little
cumbersome. Since they greatly simplify in the grand
canonical ensembles, we restrict our attention to them
and consider the Hamiltonians $H_\mu \coloneqq H - \mu N$,
where $N$ is the particle number operator and
$\mu < 0$ the chemical potential. The limit
$\mu \nearrow 0$ will be discussed later.
The corresponding thermal equilibrium states for non-interacting dynamics
are known to be quasifree, \ie they are fixed by the one-particle density
matrices \cite{RoSiTe}.
Hence also the correlation functions of the  observables in~$\fA$ 
can be expressed in terms of them \mbox{\cite[Sect.\ 3]{BaBu}}.
For example, the expectation values of the fundamental
resolvents \eqref{e.2.2}, 
forming a total set in $\fR$ \cite{BuNu}, and of their mean over
the gauge transformations, being contained in $\fA$, are given by
\be
\omega_{\beta,\mu}(R(\lambda,f)) = \int_0^\infty \! du \, e^{-u \lambda}
e^{-(u^2/2) \langle f, (e^{\beta H_\mu} - 1)^{-1} f \rangle} \, , \quad
\lambda > 0 \, , \ f \in \cD(\RR^s) \, .
\ee
Thus, in order to prove the existence of the thermodynamic limit
for non-interacting, trapped thermal systems,
it is sufficient to establish the convergence of the
corresponding one-particle
density matrices. The convergence of the expectation
functionals then follows from the dominated convergence theorem.
Their limit leads by the GNS-construction 
to a representation of the observables $\fA$ on a Hilbert space 
that describes the limit states. 

\medskip
After this outline of basic facts, we turn to the
analysis of the equilibrium states fixed by the Hamiltonians $H_R$,
defined in \eqref{e.3.1}. Since apart from the
inverse temperature $\beta > 0$ also a chemical potential
$\mu < 0$ is taken into account, the one-particle
density matrices are given by
\be \label{e.4.7} 
\omega_{\beta, \mu, R}(a^*(f) a(g)) =
\langle g, (e^{\beta (H_R - \mu)} - 1)^{-1} f \rangle
  \, , \quad f,g \in \cD(\RR^s) \, .
\ee
The existence of the thermodynamic limit of the
corresponding states on $\fA$ can then be established
with the help of Lemma \ref{l.3.1}. 
\begin{proposition} \label{p.4.1}
  Let the coupling parameters $R \mapsto c(R)$ in equation \eqref{e.3.1}
  be polynomially bounded. The states $\omega_{\beta, \mu, R}$
  converge pointwise on $\fA$ in the limit of large $R$
  to the spatially homogeneous state
  $\omega_{\beta, \mu, \infty}$, which is
  determined by the free Hamiltonian~$H_\infty$.
\end{proposition}  
\begin{proof}
As was explained, it suffices to establish the convergence of the
one-particle density matrices. According to Lemma \ref{l.3.1}, 
the unitaries $t \mapsto e^{itH_R}$ converge for large $R$ to
$t \mapsto e^{itH_\infty}$ on $L^2(\RR^s)$, 
uniformly for $t$ in compact sets. By integration, it follows that
the operators $h(H_R)$ converge on $L^2(\RR^s)$ to $h(H_\infty)$
for any Schwartz test function $h \in \cS(\RR)$.

\medskip 
Since $\mu < 0$, the function $E \mapsto (e^{\beta(E - \mu)} - 1)^{-1}$ on
the positive axis can be extended smoothly to negative values,
resulting in a Schwartz test function. As the 
Hamiltonians $H_R$ are positive, it implies that the one-particle
density matrices \eqref{e.4.7} converge for large $R$ to the
matrix determined by $H_\infty$, proving the
convergence of the states. That the limit functionals are 
spatially homogeneous is an immediate consequence of the
fact that $H_\infty$ commutes with translations, completing the proof. 
\end{proof}

We conclude that in the construction of homogeneous equilibrium states
one does not need to proceed from box approximations with specific 
boundary conditions. One can deal from the outset
with the theory in infinite space $\RR^s$ and model soft 
confining forces quite arbitrarily. Let us mention as an aside
that this applies also to the limit of macroscopic quantities, such as 
the free energy density (pressure). 

\section{Bose-Einstein condensates}
\label{sec5}
\setcounter{equation}{0}  

Having determined the thermodynamic limit states  
on $\fA$, one can turn to their analysis.
Of particular interest is the particle density in 
compact regions $\bO \subset \RR^s$
which can be calculated with the help of the corresponding particle
number operator $N(\bO)$. Picking 
an orthonormal basis $e_k$, $k \in \NN_0$, of test functions
in~$L^2(\bO)$, it is formally given by
\be
N(\bO) \coloneqq \lim_{l \rightarrow \infty} \sum_{k = 0}^l
\, a^*(e_k) a(e_k) \, .
\ee 
Its properties can be analyzed rigorously, noting that
the operators
\be
A_l(\lambda) \coloneqq (\lambda + \sum_{k = 0}^l \, a^*(e_k) a(e_k) )^{-1} 
\, , \quad \lambda > 0 \, ,
\ee
are positive elements of the observable algebra $\fA$ which are monotonically
decreasing with increasing $l \in  \NN$. Thus they converge for
large $l$ in the strong operator topology in
all representations of $\fA$. If their limit has a trivial kernel in a 
representation space, the operator $N(\bO)$ is densely defined
and can be recovered from the observables. One says in this case
that the limit state is locally normal. Technically, the
corresponding GNS-representation of the algebra $\fA(\bO)$
of observables in $\bO$ is quasi (\ie up to multiplicities)
equivalent to the Fock representation, cf.\ \cite{DADoRu}.
The other possibility 
is that the limit vanishes, which indicates an
infinite number of particles in $\bO$. This
is expected to happen in interacting theories with
attractive two-body forces, where the states collapse in
the thermodynamic limit. But also in case of repulsive or absent forces it can 
happen that the local particle number diverges if a  
particular single particle state is infinitely occupied.
It would describe an extreme case of condensation,
named proper condensate in \cite{BaBu,Bu3}. In order to
decipher these cases, a more detailed analysis is required. 

\medskip
The thermal limit states $\omega_{\beta,\mu,\infty}$,
constructed in Sect.~\ref{sec4}, have finite local
particle numbers and hence are locally normal for given
$\beta > 0$ and $\mu < 0$. 
In low dimensions, $s=1,2$, the particle density tends to
infinity if $\mu \nearrow 0$; yet if $s \geq 3$, it stays bounded
for any given temperature. In the latter case this
feature is taken as an indication of the onset of Bose-Einstein condensation.
On the other hand, it is often argued that there is no condensation in low
dimensions. Our algebraic approach sheds new light on these
issues. 

\medskip
A key element in the description of condensates, \ie multiple tensor
products of a particular one-particle state, is based on the observation
that such states often hardly perturb stationary states.
So these condensates can coexist with them for a long time.
In simple cases, such inert 
single particle states are eigenstates of the Hamiltonian. But also in the
presence of interaction, such states can appear as significant
contributions to specific many particle states.
This simple picture breaks down, however, in the thermodynamic
limit, where condensates must be described by distributions.

\medskip
In the present model, the trapped ensembles
can be complemented by coherent configurations of 
eigenstates of $H_R$, describing condensates.
Given a normalized eigenfunction
$h_R \in L^2(\RR^s)$ for the eigenvalue $\epsilon_R > 0$, 
these coherent configurations are generated 
by adjoint action of the 
unitary operators $e^{i \kappa (a^*(h_R) + a(h_R))}$, $\kappa \in \RR$,  
(Weyl automorphisms). One can perturb states by composing
them with these automorphisms.
The occupation numbers of $h_R$ can then be determined
by the number operator $N(h_R) \coloneqq a^*(h_R)a(h_R)$.
Since the Weyl automorphisms
leave the resolvent algebra $\fR$ invariant,
this step amounts to perturbing in the original state the underlying
Hamiltonian~$H_R$ by adding multiples of $(a^*(h_R) + a(h_R))$.
Thus it no longer commutes with~$N(h_R)$ and describes the
coupling of the system with a reservoir of particles with
wavefunction $h_R$.
Note that the algebra of observables $\fA$ is
not stable under the action of the perturbed dynamics. 
But if one composes a gauge-invariant 
quasifree state with the Weyl automorphisms, the
resulting state, restricted to $\fA$, is still
quasifree. It is fixed by the one-particle density matrix
in the original state and gauge invariant products of
one-point functions.

\medskip 
Returning to the quasifree and gauge-invariant
equilibrium states $\omega_{\beta, \mu, R}$, one obtains
for the one-particle density matrix in the
perturbed states in an obvious notation, $\kappa \in \RR$, 
\be \label{e.5.3}
\omega_{\beta, \mu, \kappa, R}(a^*(f)a(g)) = 
  \langle g , (e^{\beta(H_R - \mu)} - 1)^{-1} f \rangle +
  \kappa^2 \,   \langle h_R , f \rangle  \langle g, h_R \rangle  \, . 
\ee
The first term on the right hand side describes the thermal
cloud and the second one, \viz \ the gauge-invariant product of
one-point functions, describes the condensate. Note that the
chemical potential $\mu < 0$ regulates the particle density
only in the thermal cloud. Irrespective of the spatial dimension $s$,
the total particle density can in general 
be made arbitrarily large by adjusting~$\kappa$.
The resulting state on $\fA$ is stationary
for the original dynamics $\alpha_R$; as a matter of fact,
it is a non-equilibrium steady state in the sense of
\mbox{\cite[Cor.\ 3.4]{Ru}}. Moreover, it 
satisfies the KMS-condition 
on the subalgebra $\fA_{h_R}^\perp$ which is generated by the resolvents 
$A(\lambda, f^\perp)$, where
$f^\perp$ lies in the subspace of functions in the
orthogonal complement of $h_R$.

\medskip
Proceeding to the limit of large $R$, the
first term on the right hand side  of equation \eqref{e.5.3}
converges according to Proposition \ref{p.4.1}.
Since the functions $h_R$ are solutions of the equation
$\Delta h_R + \epsilon_R h_R = 0$ in the interior of the 
ball $|\bx| \leq R$, the second term would vanish in this limit 
unless $h_R$ is renormalized by a proper choice 
of $\kappa$. Since $\epsilon_R$ tends to $0$,
the resulting functions $R \mapsto \kappa_R \, h_R$
then tend to a distribution carrying zero energy,
\ie for $f \in \cD(\RR^s)$ one has 
\be
\lim_{R \rightarrow \infty} \kappa_R \, \langle h_R, (e^{itH_R} - 1) f \rangle
= \lim_{R \rightarrow \infty} (e^{it \epsilon_R} -1) \, \kappa_R \langle h_R, f \rangle
= 0 \, .
\ee
It follows that the limit states $\omega_{\beta,\mu,\kappa,\infty}$  
satisfy the KMS-condition for the limit dynamics~$\alpha_\infty$
on $\fA$ if $\mu = 0$,
\ie they are equilibrium states.
The resulting thermal cloud is homogeneous
according to Proposition~\ref{p.4.1}, but this need not be so in
case of the second term, describing
the condensate. This fact will be exemplified in the
subsequent section. Since the limit dynamics
commutes with translations, it means that the symmetry
with respect to translations is
spontaneously broken in such limit states.

\section{Spatial structure of condensates} 
\label{sec6}
\setcounter{equation}{0}

In this section, we study the properties of the
condensates appearing in the thermodynamic
limit. We first consider the case
of $s=1$ dimension. It is of some interest
since in addition to the known proper condensate, 
which emerges in the limit of vanishing chemical potential, 
there appears another condensate. We then
turn to $s=3$ dimensions and exhibit 
spatially inhomogeneous condensates. As we will
see, these unusual examples of condensates
form when the confining trap is overfilled 
with particles. 

\medskip
In $s=1$ dimension the Hamiltonians $H_R$ have discrete, simple eigenvalues
\mbox{\cite[Ch.\ 2]{BeSh}}. In view of the symmetry
properties of the potential, the eigenfunctions are either even 
or odd under spatial reflections. The (renormalized)
restrictions of the even eigenfunctions
$h_R$ of $H_R$ to the interval $|x| \leq R$ are 
$x \mapsto \cos(\sqrt{\epsilon_R} \,  x)$, and the respective 
odd eigenfunctions are 
$x \mapsto \sin(\sqrt{\epsilon_R} \, x)/\sqrt{\epsilon_R}$, 
where $\epsilon_R > 0$ is the eigenvalue.
The energy $\epsilon_{R,0} > 0$ of the ground state,
being even, can be estimated with trial functions in
the interval $|x| < R$, giving $\epsilon_{R,0} \leq \pi/4R^2$.
For test functions $f \in \cD(\RR)$, having compact
support, one then obtains for large $R$
\begin{align}
& R \mapsto \int \! dx \, h_{R,0}(x) f(x) \nonumber \\
& =  \int_{-R}^R  \! dx \, \cos(\sqrt{\epsilon_{R,0}} \, x) f(x)
= \int \! dx \, f(x) + O(1/R^2) \, .
\end{align}
Similarly, the first (renormalized) excited state $h_{R, 1}$ is odd and
its energy can be  estimated by $\epsilon_{R,1} \leq \pi/R^2$,
giving
\begin{align}
& R \mapsto \int \! dx \, h_{R,1   }(x) f(x) \nonumber \\
  & =  \int_{-R}^R  \! dx \, 
  \big( \sin(\sqrt{\epsilon_{R,1}} \, x)/\sqrt{\epsilon_{R,1}} \big) f(x) =
\int \! dx \, x f(x) + O(1/R^2) \, .
\end{align}
Thus one obtains for the thermodynamic limit of the
one-particle density matrices \eqref{e.5.3} 
in these even, respectively odd, cases the mean particle densities 
\be
x \mapsto \omega_{\beta,\mu,\kappa,\infty}(a^*(x)a(x))
= \int \! dp \, (e^{\beta(p^2 - \mu)} - 1)^{-1} +
\begin{cases}
  \kappa^2     & \text{even} \\
  \kappa^2 x^2 & \text{odd} \, .
\end{cases}  
\ee
Whereas in the even case the state is spatially homogeneous,
this symmetry is spontaneously broken in the odd case by the condensate.
In both cases the resulting states on $\fA$ are locally
normal. 

\medskip
The cause of the formation of these
different forms of condensates 
clarifies by computing the number of particles
in the interval $|x| \leq R$ in the approximating states. It is
given by 
\begin{align}
& \int_{-R}^R  dx \,  \omega_{\beta,\mu,\kappa,R}(a^*(x)a(x)) \nonumber \\
& = \int_{-R}^R dx \,  \omega_{\beta,\mu,0,R}(a^*(x)a(x)) +
\begin{cases}
  \kappa^2 \int _{-R}^R  dx \, \cos^2(\sqrt{\epsilon_R} \, x)  & \text{even} \\
  \kappa^2 \int _{-R}^R  dx \, 
  \sin^2(\sqrt{\epsilon_R} \, x) / \epsilon_R  & \text{odd} \, .
\end{cases}
\end{align} 
The first term is the mean particle number in the thermal cloud
within the interval. This number can be adjusted by choosing
$\mu < 0$. The number of particles in the condensate grows like $R$ 
in the even case, \ie the density is constant for large
$R$. In the odd case it grows like $R^3$, so this type of
condensate  appears if the confining trap is
overfilled with particles. The condensate then separates into 
two distant clusters which enclose the thermal cloud.

\medskip
These observations have consequences for the equilibrium states
in the limit of vanishing chemical potential. One obtains for the
resolvents in equation~\eqref{e.2.3} 
\be
\lim_{\mu \nearrow 0} \omega_{\beta, \mu, \kappa, \infty}(A(\lambda,f)) =  0
\quad \text{iff} \quad \int \! dx f(x) \neq 0 \, .
\ee
It shows that the particle density in the limit states is 
infinitely large, which agrees with common knowledge.
The subspace of test functions \mbox{$\cD_0(\RR) \subset \cD(\RR)$}, 
satisfying $\int \! dx  f(x) \!  = 0$, is sensitive
to the thermal cloud. It has co-dimension~$1$,
which shows that a proper condensate described by
a constant function appears in the limit. 
Since this condensate has infinite density, the limit state
does not change if one composes the approximating states
with Weyl automorphisms involving the 
ground state wavefunction $h_{R,0}$, cf.\ 
\cite[Prop.\ 3.3]{BaBu}.~Yet if one
chooses the first excited state $h_{R,1}$, the
resulting one-point functions $\int \! dx \, x f(x)$
do not vanish in general on testfunctions
$f \in \cD_0(\RR)$.
So the limit state changes and describes
a condensate which is visible on the 
subalgebra of observables $\fA_0 \subset \fA$
generated by resolvents \eqref{e.2.3}
with $f \in \cD_0(\RR)$.

\medskip
Turning to $s = 3$ dimensions, the 
eigenfunctions $h_R$ of the Hamiltonian~$H_R$
for eigenvalues $\varepsilon_R > 0$
are given by products of
spherical harmonics and solutions of the radial
Schr\"odinger equation. The restrictions of the radial 
solutions to the region $|\bx| \leq R$ are
Bessel functions, which depend on $\epsilon_R$ and
the angular momentum $l$. Whereas the
renormalized functions for $l=0$
(which include the ground state) lead to constant functions
in the limit of large $R$, the solutions for $l \in \NN$
give rise to non-homogeneous condensates. To illustrate
this fact, we consider the case $l=1$. There the 
renormalized functions in the region $|\bx| \leq R$  are given in 
coordinates $\bx = (x,y,z)$ by 
\be
\bx \mapsto h_{\epsilon_R}(\bx) =
z \, (k_R \, |\bx|)^{-3} \big( (k_R |\bx|) \cos(k_R |\bx|)
- \sin(k_R |\bx|) \big) \, ,
\ee
where $k_R \coloneqq \sqrt{\epsilon_R}$.
The energy of the lowest
eigenvalue can be estimated as in $s=1$ dimension, giving the upper
bound $k_R \leq 3 \pi/(2 R)$.  It follows that
\be
|h_{\epsilon_R}(\bx) - z| \leq c |z| |\bx|^2/R^2 \quad
\text{for} \ |\bx| \leq R \, ,
\ee
where the constant $c$ does not depend on $R$. Hence
for any test function $f \in \cD(\RR^3)$ one arrives at
\be
R \mapsto \int \! d \bx \, h_R(\bx) f(\bx) =
\int \! d \bx \, z   f(\bx) + O(1/R^2) \, ,
\ee
resulting in an inhomogeneous condensate in the thermodynamic
limit. Computing the number of particles in  
the corresponding equilibrium states in  
bounded regions, one finds again that the approximating
states must be overfilled with particles
in order to arrive at such inhomogeneous condensates.
Nevertheless, the corresponding equilibrium states
are all locally normal. 

\medskip 
To summarize, 
in one and two dimensions the particles in the thermal cloud
form a proper homogeneous condensate with infinite
density for vanishing chemical potential. The cloud remains  
visible on a subalgebra of observables. This algebra is still
large (it acts irreducibly on all subspaces of Fockspace
with a finite particle number). Adding an 
inhomogeneous condensate to the equilibrium states, 
it can be recognized by these observables.
In three and more dimensions the states necessarily have to form
condensates if a specific density of the cloud is exceeded.
Depending on the filling of the regions, the resulting condensates
can be homogeneous or inhomogeneous.

\medskip
We conclude this section with a remark on the detection of condensates.
In homogeneous systems, a standard tool is the analysis of 
spatial correlations between the basic
creation and annihilation operators that enter the observables,
\ie the search for ``off-diagonal long range order''. In case of
inhomogeneous systems, such as the preceding ones, this strategy
leads to problems in view of the growth of the condensates at
large distances. Since the states are stationary, it is more
meaningful to look at temporal correlations. This is possible
since the algebra is stable under the dynamics. 

\medskip
For $s \geq 3$, one proceeds from the one-particle density matrices,   
\begin{align}
& t \mapsto \omega_{\beta, 0, \kappa, \infty}(\alpha_\infty(t)(a^*(f)) \, a(g))
\nonumber \\
& = \langle g, (e^{\beta H_\infty} - 1)^{-1} e^{itH_\infty} f \rangle
+ \kappa^2 \langle h_\infty, e^{it H_\infty} f \rangle
\langle g, h_\infty \rangle \, .  
\end{align}
The first term on the right hand side converges to $0$ in the limit
of large $t$ since the spectrum of $H_\infty$ is absolutely continuous. 
The second term is constant in~$t$ since $h_\infty$ is a distribution
of zero energy. Thus  one arrives at
\be
\lim_{t \rightarrow \infty}
\omega_{\beta, 0, \kappa, \infty}(\alpha_\infty(t)(a^*(f)) \, a(g))
= \kappa^2 \langle h_\infty, f \rangle
\langle g, h_\infty \rangle \, .  
\ee
Whereas for $\kappa = 0$ the temporal correlations vanish at
large time distances, this is not the case in presence of condensates. 
The memory of them is not lost. This applies also to the temporal
correlations between arbitrary elements of the resolvent algebra $\fR$. 
In this context, it is important that each observable commutes
at asymptotic times with all observables at any given time
(\ie the dynamics acts in an asymptotically abelian manner). 
Thus, the time translated observables form central sequences;
their limits have the
meaning of generalized charge operators that discriminate disjoint
representations, \eg  between thermal
clouds with and without condensate. 
This will also hold in the presence of interaction.

\section{Conclusions}
\label{sec7}
\setcounter{equation}{0}

The  algebraic approach to many body physics is of advantage when studying
these systems in the limit of infinite numbers of particles and infinite
volume.  There one leaves the comfortable framework of the Stone-von Neumann
theorem and is confronted with a vast number of disjoint
Hilbert space representations
of the algebra of observables. To identify the correct representation
spaces, one resorts to limits of expectation 
functionals over the algebra. The representation spaces can then be
determined from the limits using the GNS-construction. 

\medskip
Our approach, based on the resolvent algebra, is of particular interest
in this context, because this algebra is stable under the automorphic action
of the dynamics for a multitude of interactions. In this respect,
the resolvent algebra is superior
to the Weyl algebra, which neither contains any
gauge-invariant observable, nor
does it allow for generic two-body interactions.
Moreover, the resolvent algebra describes from the outset particles 
in infinite space. One does not rely on box approximations
with impenetrable walls. 

\medskip
To demonstrate the merits of the algebraic approach, we have conducted
in this framework a study of non-interacting particles trapped by soft
confining forces. Although similar models have been studied many times,
our algebraic approach has revealed some new aspects.
First, it was shown that the dynamics of confined systems
converges on the algebra to the free, spatially homogeneous dynamics by
removing the boundary forces in quite arbitrary ways. The 
limit dynamics is unique and does not depend on details of the
confining forces. Second, the flexibility in constructing
representations of the algebra led to the result 
that thermal equilibrium states including condensates
exist in any number of spatial dimensions and
at any temperature. The transition
temperatures in three and more dimensions, where condensation
necessarily takes place in a gas of Bosons,
do not exist in low dimensions. 
Third, we have shown that the condensates appearing in the
thermodynamic limit need not be homogeneous.
Since the limit dynamics commutes with translations, it
means that this symmetry is spontaneously broken in the
limit states. This phenomenon occurs when the
approximating systems are overfilled with particles. 
Instead of generating infinite local particle densities,
the condensate evades and forms separate clusters,
leaving behind an inhomogeneous state. 

\medskip
The algebraic approach has the potential of being useful
also in interacting theories. The dynamics of trapped
systems which are covered
in the framework of the resolvent algebra are
induced by Hamiltonians on Fock space of the form
\begin{align}
H_R & \coloneqq \int \! d\bx \, \big(\bpartial a^*(\bx) \bpartial a(\bx)
+ c_R^2 \, \theta(\bx^2 - R^2)(\bx^2 - R^2) \, a^*(\bx) a(\bx) \big) \nonumber \\
& + \iint \! d\bx d\by \ a^*(\bx) a^*(\by) V(\bx - \by) a(\bx) a(\by) \, .
\end{align}
The operator in the first line coincides with \eqref{e.3.1} 
and the operator in the second line describes a two-body interaction
with an arbitrary potential $V$ which is 
continuous and  vanishes at infinity. We assume here that $V \geq 0$,
but the formalism also covers potentials of positive type which
are repulsive at short distance and attractive
at larger distances. They may be of interest in
investigations of bosonic crystals.

\medskip
In the presence of interaction, 
the observable algebra $\fA$ acting on
spaces with a finite number of particles is likewise invariant under the 
action of the dynamics $t \mapsto \alpha_R(t)$
induced by the Hamiltonians $H_R$;
moreover, this action is pointwise
norm-continuous, cf.\ \cite[Prop.\ 4.4]{Bu2}.
The results for the untruncated harmonic trapping potential
in the appendix of \cite{Bu2} imply that by unfolding the potential
the corresponding sequence of dynamics converges
as in Lemma~\ref{l.3.3}
to the spatially homogeneous interacting dynamics; the radii
are replaced by the decreasing couplings of the harmonic potential.
Thus an analogue of Proposition \ref{p.3.4} obtains in the   
interacting case with the consequences discussed above.
It would be desirable to extend this result to the larger 
family of truncated harmonic potentials considered here. 

\medskip
The canonical and grand canonical Gibbs ensembles for the trapped
systems exist and are represented by density matrices in Fock space.
However, the determination of their thermodynamic limits is still in
its infancy.  Based on the results mentioned above, one knows that
all weak limit points are steady states for the
homogeneous limit dynamics that induce representations
of the algebra of observables on Hilbert spaces. 
However, apart from oscillating lattice systems
with nearest neighbor interactions \cite{Bu1},  it is
not known whether the limit states are still ground, respectively
equilibrium states. 

\medskip
Independently of the solution of this question, the formation of condensates
can be discussed in the present framework.
Based on the idea that specific configurations of single particle
states in condensates
do not have a strong impact on steady states
one may adopt the strategy used in the non-interacting case. 
Taking for example the ground state $\psi_n$ for the
dynamics $\alpha_{R_n}$ in the $n$-particle space,
it is a common approach to focus on the eigenvalues of
the corresponding one-particle density matrix. 
If the largest eigenvalue is proportional to the particle number~$n$,
the corresponding one-particle vector is regarded
as the building block of a condensate \cite{OnPe}.
It is then plausible to represent the pure 
condensate by the $n$-fold tensor product of this 
vector, which in general is no longer an eigenstate of the
Hamiltonian, however. Although this
approach is meaningful from the point of view of
physics, it is plagued by two technicalities.
First, if there is such a one-particle state, there
exists a dense set of one-particle states whose
expectation number is also proportional to~$n$.
So it is difficult to determine that particular state
and arrive at an unambiguous picture of the condensate.
Second, this approach breaks down in the
thermodynamic limit since then the building
blocks of the condensate are to be described
by distributions.

\medskip 
For this reason a complementary approach was proposed in
\cite{BaBu,Bu3}. Given the sequence of ground states $\psi_n$,
$n \in \NN$, one looks at the subspace of test functions
$f \in \cD(\RR^s)$, for which the expectation values of
the corresponding particle
number operators $N(f)$ stay bounded in the state in 
the limit of large $n$. If this holds for all test
functions, it means that there appears no condensate.
More interesting are those cases where the expectation
values stay bounded on a subspace of test functions.
If its co-dimension in $\cD(\RR^s)$ is equal to~$1$ it
indicates a proper condensate \cite{BaBu,Bu3}, but higher co-dimensions
may be possible. In the proper case 
one can recover a distribution describing the
condensate. As in the non-interacting case, it is  
given by a linear functional on the test functions  
that encodes the properties of the condensate.

\medskip 
In order to arrive at condensates with a 
non-trivial spatial structure, our present results suggest
that one has to proceed from states $\varphi_n$ with
energies which are large compared to the 
ground state energy (their quotient should
be bigger than $1$). For increasing
numbers, the bulk of the particles is then expected to evade
more rapidly to
large distances, so one has to renormalize
the states by suitable powers of $n \in \NN$. 
Again, the subspace of test functions for which the
corresponding particle numbers stay finite
in the limit is
expected to have co-dimension $1$ in case of
proper condensates. One then arrives at 
distributions which differ from those obtained 
for ground states.

\medskip 
Having identified these distributions one
can add the corresponding condensates in
arbitrary portions to a
steady state for the limit dynamics.
This is accomplished by composing the state with an
automorphism defined by the adjoint action
of a generalized Weyl operator involving such a 
distribution. In contrast to the free case,
the resulting state will no longer be stationary
in general since 
the automorphisms do not commute with the
dynamics. The resulting state describes 
the situation, where the original state and
a condensate were brought into contact at a given
time. An investigation of its large-time behavior
would then reveal whether the condensate does 
coexist with the steady state. 

\newpage
\noindent
    {\Large \bf Acknowledgment} \\[2mm]
    We thank Robert Seiringer for a helpful discussion. 
    DB would like to thank the Erwin Schr\"odinger Institute for 
    hospitality and financial support during the final phase of this project.

\end{document}